\documentclass[a4paper,11pt]{article}

\usepackage{tcs}
\makeatletter
\newcommand{\fallbackcref}[1]{\expandafter\fallbackcref@dispatch#1:\@nil{#1}}
\def\fallbackcref@dispatch#1:#2\@nil#3{%
    \@ifundefined{fallbackcref@#1}{\ref{#3}}{\csname fallbackcref@#1\endcsname{#3}}%
}
\def\fallbackcref@thm#1{Theorem~\ref{#1}}
\def\fallbackcref@theorem#1{Theorem~\ref{#1}}
\def\fallbackcref@conj#1{Conjecture~\ref{#1}}
\def\fallbackcref@lemma#1{Lemma~\ref{#1}}
\def\fallbackcref@lem#1{Lemma~\ref{#1}}
\def\fallbackcref@prop#1{Proposition~\ref{#1}}
\def\fallbackcref@cor#1{Corollary~\ref{#1}}
\def\fallbackcref@def#1{Definition~\ref{#1}}
\def\fallbackcref@eq#1{Equation~\eqref{#1}}
\def\fallbackcref@sec#1{Section~\ref{#1}}
\newcommand{\installfallbackcref}{%
    \renewcommand{\cref}[1]{\fallbackcref{##1}}%
    \renewcommand{\Cref}[1]{\fallbackcref{##1}}%
}
\makeatother

\IfFileExists{cleveref.sty}{
    \crefname{theorem}{Theorem}{Theorems}
    \Crefname{theorem}{Theorem}{Theorems}
    \crefname{thm}{Theorem}{Theorems}
    \Crefname{thm}{Theorem}{Theorems}
    \crefname{conjecture}{Conjecture}{Conjectures}
    \Crefname{conjecture}{Conjecture}{Conjectures}
    \crefname{definition}{Definition}{Definitions}
    \Crefname{definition}{Definition}{Definitions}
    \crefname{lemma}{Lemma}{Lemmas}
    \Crefname{lemma}{Lemma}{Lemmas}
    \crefname{proposition}{Proposition}{Propositions}
    \Crefname{proposition}{Proposition}{Propositions}
    \crefname{corollary}{Corollary}{Corollaries}
    \Crefname{corollary}{Corollary}{Corollaries}
    \crefname{equation}{Equation}{Equations}
    \Crefname{equation}{Equation}{Equations}
}{
    \installfallbackcref
}

\begin{document}

\title{Sharp Bounds on the Eigenvalues of Kikuchi Graphs \\and Applications to Quantum Max Cut}
\date{}
\author{
    \begin{tabular}{cc}
        Ainesh Bakshi & Arpon Basu \\
        \texttt{ainesh@nyu.edu} & \texttt{arpon.basu@princeton.edu} \\
        NYU & Princeton \\
        \\[-0.25em]
        Pravesh Kothari & Anqi Li \\
        \texttt{kothari@cs.princeton.edu} & \texttt{aqli@cs.washington.edu} \\
        Princeton & University of Washington
    \end{tabular}
}

\maketitle
\thispagestyle{empty}

\abstract{
We prove that the maximum eigenvalue of the (both signed and signless) Laplacian of level
$k$ Kikuchi graph of any graph $G$ with $m$ edges is at most $m+k$. This confirms four
recent conjectures of Apte, Parekh, and Sud. 

As applications, we obtain that tensor products of one and two qubit states achieve
an approximation ratio of $5/8$ for Quantum Max Cut and $5/7$ for the XY Hamiltonian.
Moreover, combining our bounds with the algorithms analyzed by Apte, Parekh, and Sud, yields
efficient algorithms achieving an approximation ratio of $0.614$ for Quantum Max Cut and
$0.674$ for the XY Hamiltonian. Finally, we also make progress over Lew's bound on Brouwer's conjecture about the sum of the top-$k$ eigenvalues of a Graph Laplacian.
}

\newpage
\setcounter{page}{1}

\section{Introduction}

Let $G$ be a graph on $[n]$ with $m$ edges. 
For $0\leq k\leq n$, the level-$k$ \emph{Kikuchi graph} of $G$, denoted $F_k(G)$, has vertex set
$\binom{[n]}{k}$. Two $k$-sets $S,T \in {[n] \choose k}$ are adjacent precisely when
their symmetric difference is an edge of $G$, i.e. $S\oplus T = \{i,j\}$ for some edge 
$\{i,j\}$ of $G$. In the combinatorics literature, the Kikuchi graph is often called the \emph{token graph},
which can be interpreted as the graph of configurations of $k$ indistinguishable tokens placed on 
distinct vertices of $G$. A valid move between two configurations is to slide a \emph{token} across
an edge of $G$ to an unoccupied vertex.  
The Kikuchi perspective is useful because it turns questions about high-dimensional incidence
structure into questions about walks, eigenvalues, and cycles in an ordinary graph. This
viewpoint has led to the proof of the Hypergraph Moore
Bound conjectured by Feige~\cite{Feige08} and refutation algorithms for random constraint satisfaction
~\cite{GuruswamiKM22, HsiehKM23}, as well as exponential lower bounds for locally
correctable codes~\cite{KothariM23}. 

Yet another connection emerges from the quantum information literature, where Apte, Parekh, and
Sud~\cite{AptePS25} observed that the maximum energies of several well-studied $2$-local Hamiltonians 
can be expressed in terms of extremal eigenvalues of Kikuchi graphs. In particular, after decomposing
the $n$ qubit Hilbert space into low Hamming weight sectors, the action of these Hamiltonians on the 
weight $k$ sector is determined by the adjacency and Laplacian matrices associated to the level 
$k$ Kikuchi graph. Thus, sharp spectral bounds for Kikuchi graphs translate directly into 
sharp upper bounds on the maximum energies of these local Hamiltonians. This observation led Apte, 
Parekh, and Sud to propose a set of sharp spectral conjectures for Kikuchi graphs:

\begin{conjecture}[Apte--Parekh--Sud Conjectures~\cite{AptePS25}]
\label{conj:aps}
Given a graph $G = ([n], E)$, let $F_k(G)$ be the associated level $k$ Kikuchi graph. Let $A^{(k)}_G$ be
the adjacency matrix of $F_k(G)$, and let $L^{(k)}_G$ and $Q^{(k)}_G$ be the signed and signless 
Laplacian respectively. Then, for every graph $G=([n],E)$ with $m=\abs{E}$, and every $0\leq k\leq n$,
\[
\begin{array}{r@{\;}c@{\;}l@{\hspace{1.2em}}l@{\qquad\qquad}r@{\;}c@{\;}l@{\hspace{1.2em}}l}
    \lambda_{\max}(L^{(k)}_G) &\leq& m + k
        & \textup{(C1)}
        &
    \lambda_{\max}(Q^{(k)}_G) &\leq& m + k
        & \textup{(C2)}
    \\[0.5em]
    \lambda_{\max}(A^{(k)}_G) &\leq& \frac{1}{2}\Paren{m + k}
        & \textup{(C3)}
        &
    \lambda_{\min}(A^{(k)}_G) &\geq& -\frac{1}{2}\Paren{m + k}
        & \textup{(C4)}.
\end{array}
\]
\end{conjecture}

In this paper, we resolve the four APS conjectures \textup{(C1)}--\textup{(C4)}. 
\begin{theorem}
\label{thm:mainthm}
    For any graph $G = ([n], E)$ with $|E| = m$ edges, and any $0\leq k\leq n$, we have 
    \[
        -(m + k)\leq -\lambda_{\max}(L^{(k)}_G)
        \leq 2\lambda_{\min}(A^{(k)}_G)
        \leq 2\lambda_{\max}(A^{(k)}_G)
        \leq \lambda_{\max}(Q^{(k)}_G)\leq m + k.
    \]
\end{theorem}

Prior to this work, the best known general bound in this direction was Lew's estimate 
for the signless Laplacian: $\lambda_{\max}(Q_G^{(k)})\leq m+4k-2$~\cite{Lew26}. Further, 
each of the bounds in \cref{thm:mainthm} is tight. It is easy to check that a disjoint union of 
$k$ stars saturates the bounds for the largest eigenvalue of the signed and signless Laplacian. 
For the adjacency matrix, consider the special case of a disjoint union of $k$ edges, i.e. $G$ 
is a perfect matching on $2k$ vertices. Then, the level $k$ Kikuchi graph has one token on each base 
edge, and the configurations can be represented as bit strings on $\braces{0,1}^k$, denoting whether 
the token is on the right or left vertex. The Kikuchi graph has a connected component isomorphic to the $k$-dimensional hypercube,
whose adjacency matrix satisfies $\lambda_{\max} = k$ and $\lambda_{\min}=-k$.

\paragraph{Applications.} Local Hamiltonians are the basic mathematical model for quantum many-body
systems. A Hamiltonian is a Hermitian operator whose eigenvalues are the possible energies
of the system. In physically relevant models, it is typically a sum of simple terms, each
involving only a few particles. The central computational task is to estimate an extremal
energy: the ground-state energy, or, after negating the Hamiltonian, the maximum energy.
From the perspective of theoretical computer science, this is the quantum analogue of
constraint satisfaction, containing classical problems such as Max Cut and $3$-SAT as
special cases. We focus on $2$-local Hamiltonians of the form
\[
    H = \sum_{\{i,j\}\in E} w_{i,j} h_{i,j}, \qquad w_{i,j}>0,
\]
where $h_{i,j}$ acts non-trivially only on particles $i$ and $j$. The maximum energy is
$\lambda_{\max}(H)=\max_{\|\psi\|=1}\langle \psi,H\psi\rangle$. 
We consider the three families of $2$-local Hamiltonians that have been extensively studied,
the Quantum Max Cut (QMC) Hamiltonian, the EPR Hamiltonian, and the XY Hamiltonian:
\[
\begin{gathered}
    H_{\mathrm{QMC}} =
        \sum_{\{i,j\}\in E} \tfrac{1}{2}\Paren{I_i I_j - X_iX_j - Y_iY_j - Z_iZ_j}, \\
    \begin{alignedat}{2}
    H_{\mathrm{EPR}} &=
        \sum_{\{i,j\}\in E} \tfrac{1}{2}\Paren{I_i I_j + X_iX_j + Y_iY_j - Z_iZ_j},
    \qquad
    H_{\mathrm{XY}} &=
        \sum_{\{i,j\}\in E} \Paren{I_i I_j - \tfrac{1}{2}\Paren{X_iX_j + Y_iY_j}}.
    \end{alignedat}
\end{gathered}
\]
Here $X_i$ denotes the operator acting as $X$ on qubit $i$ and as the identity
on all other qubits, and similarly for $Y_i,Z_i,I_i$ and together they form the set of
Pauli matrices:
\[
    I=\begin{pmatrix}1&0\\0&1\end{pmatrix}, \qquad
    X=\begin{pmatrix}0&1\\1&0\end{pmatrix}, \qquad
    Y=\begin{pmatrix}0&-i\\i&0\end{pmatrix}, \qquad
    Z=\begin{pmatrix}1&0\\0&-1\end{pmatrix}.
\]

As a direct consequence of \cite[Corollaries 5-10]{AptePS25}, we obtain improved 
approximation guarantees for the QMC, EPR, and XY problems.
\begin{theorem}[Consequences for local Hamiltonians]
\label{thm:applications}
\label{thm:stuctural}
\label{thm:algorithmic}
There exist tensor products of one- and two-qubit states achieving the structural
ratios below. Moreover, the ratios in the final column are achieved by efficient
algorithms that prepare such states.
\[
\begin{array}{c@{\qquad}c@{\qquad}c}
\toprule
\text{Hamiltonian} & \text{Existence} & \text{Efficient algorithm} \\
\midrule
H_{\mathrm{QMC}} & 5/8 & 0.614 \\
H_{\mathrm{XY}} & 5/7 & 0.674 \\
H_{\mathrm{EPR}} & (\sqrt{5}+1)/4 & (\sqrt{5}+1)/4 \\
\bottomrule
\end{array}
\]
\end{theorem}

\paragraph{Comparison to prior work.} Quantum Max Cut was formulated by Gharibian and
Parekh~\cite{GharibianP19}, who gave a $0.498$-approximation for QMC and a
$0.649$-approximation for XY. For QMC on general graphs, subsequent worst-case
ratios progressed through $1/2$~\cite{ParekhT22}, $0.526$~\cite{HuberTPG24},
$0.531$~\cite{anshu2020beyond}, $0.533$~\cite{parekh2021application},
$0.562$~\cite{Lee22}, $0.595$~\cite{LeeP24}, $0.599$~\cite{JorqueraKKSW24},
and $0.603$~\cite{GriblingSS25}. The previous best was $0.611$, due to Apte,
Lee, Marwaha, Parekh, and Sud~\cite{ApteLMPS25}. EPR was introduced by
King~\cite{King23}; its ratios progressed from $1/\sqrt{2}$~\cite{King23} to
$18/25$~\cite{JorqueraKKSW24}, $(1+\sqrt{5})/4$~\cite{JuN25,ApteLMPS25}, and
finally $0.8395$, due to Apte, Lee, Marwaha, Parekh, Sinjorgo, and
Sud~\cite{ApteLMPSS25}.

Our QMC and XY guarantees improve the best known worst-case ratios. For EPR, our
ratio does not beat $0.8395$, but it gives a simpler certificate: a tensor product
of one- and two-qubit states. This is a much smaller ansatz than the AGM circuits
used in several prior works~\cite{anshu2020beyond,King23,ApteLMPS25,ApteLMPSS25}:
a block-product state can entangle each qubit with at most one partner, while AGM
circuits apply commuting two-qubit rotations across many edges and can exploit
overlapping local correlations. Exact optimization is QMA-complete for QMC and
XY~\cite{PiddockM17}, while EPR lies in StoqMA~\cite{PiddockM17}. QMC is also
APX-hard~\cite{Piddock25}, and under a conjectured vector-valued Borell
inequality, obtaining a $(0.956+\epsilon)$-approximation is as hard as the Unique
Games Conjecture~\cite{HwangNP0W23}.

\paragraph{Brouwer's Conjecture.}
Brouwer's conjecture is a central open problem on partial sums of Laplacian
eigenvalues. For a graph $G=([n],E)$, let $\epsilon_r(G)=\sum_{i=1}^r \lambda_i(L(G))-\abs{E}$. 
Brouwer~\cite{BrouwerH12,HaemersMTR10} conjectured that for every graph $G$ and every $1\leq r\leq n$,
$\epsilon_r(G)\leq \binom{r+1}{2}$. 

\begin{theorem}[Improved estimates for Brouwer's conjecture]\label{thm:Brouwer}
Let $G = ([n],E)$ be a graph with $m$ edges, and order the Laplacian eigenvalues as
$\lambda_1(L(G))\geq\cdots\geq\lambda_n(L(G))$.
For every $1\leq r\leq n$, we have
\[
    \epsilon_r(G)
    \leq \binom{r+1}{2}+\frac{4}{3}r^{3/2}-r+\sqrt{r}.
\]
\end{theorem}

\Cref{thm:Brouwer} improves a recent estimate of Lew~\cite{Lew26}, who showed
that $\epsilon_r(G) \leq \binom{r+1}{2} + 4 r^{3/2} -r - 2 \sqrt{r}$. The asymptotic improvement is 
a factor of $1/3$ on the leading term. After our paper was posted online,~\cite{KothariT26} resolved Brouwer's conjecture completely, using different methods. We retain our proof to highlight the general applicability of our techniques beyond just the proof of \cref{thm:mainthm}.

\section{Notation and Preliminaries}
For any two sets $X, Y$, define their symmetric difference $X\oplus Y$ to be
$(X\setminus Y)\cup(Y\setminus X)$. It is easy to see that $X\oplus Y = Y\oplus X$,
and if $X\oplus Y = Z$, then $X\oplus Z = Y$.

\begin{definition}[Kikuchi Graphs]
\label{def:kikuchi}
Let $G = ([n],E)$ be a graph with $m = \abs{E}$ edges. For any $k \leq n$, the
\emph{level $k$ Kikuchi graph} $F_k(G)$ of $G$ is a graph on $\binom{[n]}{k}$
with edges $S \sim T$ if and only if $S \oplus T \in E$. 
\end{definition}
We denote the Laplacian of $F_k(G)$ by $L_G^{(k)}$. Also denote by $A_G^{(k)}$ the
adjacency matrix of $F_k(G)$, and let
$D_G^{(k)}\in\R^{\binom{[n]}{k}\times\binom{[n]}{k}}$ be the \emph{degree matrix} of
$F_k(G)$, i.e. $D_G^{(k)}$ is a diagonal matrix whose rows and columns are indexed by
elements of $\binom{[n]}{k}$, with $D_G^{(k)}(S, S)$ being defined to be the degree of
the vertex $S\in\binom{[n]}{k}$ in the graph $F_k(G)$. 

Note that $L_G^{(k)} = D_G^{(k)} - A_G^{(k)}$. We also define the \emph{signless}
Laplacian $Q_G^{(k)}$ to be $Q_G^{(k)}:= D_G^{(k)} + A_G^{(k)}$. It is a standard fact
that $L_G^{(k)}, Q_G^{(k)}$ are positive semidefinite matrices, i.e.
$L_G^{(k)}\succcurlyeq 0, Q_G^{(k)}\succcurlyeq 0$. Here $\succcurlyeq$ refers to the
Loewner order on Hermitian matrices, with $A\succcurlyeq B$ denoting that $A - B$ is
positive semidefinite. One consequence of positive semidefiniteness is the following:
\begin{proposition}\label{proposition:ALQ-relationship}
    $-\frac{1}{2}\cdot\lambda_{\max}(L_G^{(k)})\leq\lambda_{\min}(A_G^{(k)})
    \leq\lambda_{\max}(A_G^{(k)})\leq\frac{1}{2}\cdot\lambda_{\max}(Q_G^{(k)})$. 
\end{proposition}
\begin{proof}
    Note that $Q_G^{(k)} - L_G^{(k)} = 2A_G^{(k)}$, and thus 
    \[
        A_G^{(k)} = \frac{1}{2}\Paren{Q_G^{(k)} - L_G^{(k)}}
        \overset{(\ast)}{\preceq}\frac{1}{2}\cdot Q_G^{(k)}
        \implies
        \lambda_{\max}(A_G^{(k)})\leq\frac{1}{2}\cdot\lambda_{\max}(Q_G^{(k)}),
    \]
    where $(\ast)$ uses the fact that $L_G^{(k)}\succeq 0$. Similarly, we also have 
    \[
        A_G^{(k)} = \frac{1}{2}\Paren{Q_G^{(k)} - L_G^{(k)}}
        \overset{(\ast)}{\succeq}-\frac{1}{2}\cdot L_G^{(k)}
        \implies
        \lambda_{\min}(A_G^{(k)})\geq-\frac{1}{2}\cdot\lambda_{\max}(L_G^{(k)}).
        \qedhere
    \]
\end{proof}
Thus to establish \cref{thm:mainthm} it suffices to show that
$\lambda_{\max}(L_G^{(k)}), \lambda_{\max}(Q_G^{(k)})\leq m + k$.

\section{Eigenvalue Bounds for Kikuchi Graphs}

In this section, we prove the Laplacian and signless Laplacian bounds in the
Apte--Parekh--Sud conjectures.
It suffices to prove these two bounds: once the signless Laplacian bound is known,
the adjacency bound in \cref{conj:aps} follows from
\cref{proposition:ALQ-relationship} and
$\rho(Q_G^{(k)})=\lambda_{\max}(Q_G^{(k)})$.
We treat the two cases simultaneously.
For $b\in\{-1,1\}$, define
\[
    L_{b,G}^{(k)}=D_G^{(k)}+b A_G^{(k)}.
\]
Thus $L_{-1,G}^{(k)}=L_G^{(k)}$ and $L_{1,G}^{(k)}=Q_G^{(k)}$.
The main theorem of this section is the following.
\begin{theorem}
\label{thm:token-lap-bound}
For any graph $G=([n],E)$ with $m=\abs{E}$, any $0\leq k\leq n$, and any
$b\in\{-1,1\}$,
\begin{equation*}
    \lambda_{\max}(L_{b,G}^{(k)})\leq m+k.
\end{equation*}
\end{theorem}

\paragraph{Setting up definitions.}

For an edge $f=\{i,j\}\in E(G)$, let
\begin{equation}
\label{eq:token-edge-fiber}
    \calE_f^{(k)}
    =
    \Braces{
        \{S,S\symdiff f\}: S\in \binom{[n]}{k},\, \abs{S\cap f}=1
    }
\end{equation}
be the set of token-graph edges obtained by moving a token across the fixed edge $f$ of
the base graph.
The edge set of $F_k(G)$ is therefore
\begin{equation}
\label{eq:token-edge-decomp}
    E(F_k(G))=\bigcup_{f\in E(G)}\calE_f^{(k)}.
\end{equation}
For $f\in E(G)$, let $L_{b,f}^{(k)}$ be the $D+b A$ operator of the graph
with vertex set $\binom{[n]}{k}$ and edge set $\calE_f^{(k)}$. Then
\[
    L_{b,G}^{(k)}=\sum_{f\in E(G)} L_{b,f}^{(k)}.
\]
Equivalently, for every function $x$ on $\binom{[n]}{k}$,
\begin{equation}
\label{eq:local-token-lap}
    (L_{b,f}^{(k)} x)_S =
    \begin{cases}
        x_S+b x_{S\symdiff f}, & \text{if } \abs{S\cap f}=1,\\
        0, & \text{otherwise.}
    \end{cases}
\end{equation}

The proof proceeds by induction on $k$.
The key idea is to differentiate an eigenvector across an edge of the original graph.
For $b=-1$, this is the usual difference across an edge; for $b=1$, it is
the corresponding signless edge derivative. Edges disjoint from the differentiated edge
produce a $(k-1)$-token operator on the deleted graph, while adjacent edges produce signed
coordinate permutations of neighboring edge derivatives.

\begin{proof}
Fix $b\in\{-1,1\}$. 
The base case $k=0$ is immediate: the $0$-token graph has one vertex, so
$L_{b,G}^{(0)}=0$.
The case $k=n$ is also trivial, since the $n$-token graph has one vertex.
We therefore assume $1\leq k\leq n-1$, and assume the theorem has been proved for
$k-1$ for every simple graph, with this fixed value of $b$.
Let $x\in\R^{\binom{[n]}{k}}$ be an eigenvector of $L_{b,G}^{(k)}$ with eigenvalue
$\lambda$,
\begin{equation}
\label{eq:eigenvector-equation}
    L_{b,G}^{(k)}x=\lambda x.
\end{equation}
Since
\[
    \langle y,L_{b,G}^{(k)}y\rangle
    =
    \sum_{\{S,T\}\in E(F_k(G))}
    \Paren{y_S+b y_T}^2,
\]
the operator $L_{b,G}^{(k)}$ is positive semidefinite, and hence $\lambda\geq 0$.
We will show $\lambda\leq m+k$.
Orient every edge of $G$ arbitrarily.
For an oriented edge $e=(i,j)$, define the background configuration $D_e$ to be all the
possible ways to place $k-1$ tokens away from the edge $e$:
\begin{equation*}
    D_e=\Braces{R\subseteq [n]\setminus\{i,j\}: \abs{R}=k-1} = \binom{[n]\setminus\{i,j\}}{k-1}.
\end{equation*}
Next, consider the edge gradient vector
$g_e\in \R^{D_e}\subset\R^{\binom{[n]}{k-1}}$\footnote{We extend a vector in
$\R^{D_e}$ to a vector in $\R^{\binom{[n]}{k-1}}$ by placing zeros in the coordinates
in $\binom{[n]}{k-1}\setminus D_e$.} such that for $R\in D_e$,
\begin{equation}
\label{eq:edge-derivative}
    g_e(R)=x_{R\cup\{i\}}+b x_{R\cup\{j\}}.
\end{equation}
Also consider a similar gradient, $h_e$, obtained after applying the token operator:
\begin{equation}
\label{eq:lap-edge-derivative}
    h_e(R)=(L_{b,G}^{(k)}x)_{R\cup\{i\}}
    +b (L_{b,G}^{(k)}x)_{R\cup\{j\}}.
\end{equation}
Using \eqref{eq:eigenvector-equation}, we have
\begin{equation}
\label{eq:h-g}
    h_e=\lambda g_e.
\end{equation}

We first compute $h_e$ in terms of the edge derivatives $g_f$.
Let $e=(i,j)$ be fixed and write
\[ S=R\cup\{i\},\qquad T=R\cup\{j\}. \]
There are three types of edges $f\in E(G)$. First, if $f=e$, then
\begin{align*}
    (L_{b,e}^{(k)}x)_S+b (L_{b,e}^{(k)}x)_T
    &= (x_S+b x_T)+b(x_T+b x_S)  \\
    &=2(x_S+b x_T) =2g_e(R).
\end{align*}

Second, suppose that $f$ is disjoint from $e$.
If $\abs{R\cap f}\neq 1$, there is no legal token move across $f$ from either $S$ or $T$,
and the contribution is zero.
If $\abs{R\cap f}=1$, then
\[  S\symdiff f=(R\symdiff f)\cup\{i\}, \qquad
    T\symdiff f=(R\symdiff f)\cup\{j\}. \]
Therefore,
\begin{align*}
    (L_{b,f}^{(k)}x)_S+b (L_{b,f}^{(k)}x)_T
    &= \Paren{x_{R\cup\{i\}}+b x_{(R\symdiff f)\cup\{i\}}}
    +b\Paren{x_{R\cup\{j\}}+b x_{(R\symdiff f)\cup\{j\}}} \\
    &= \Paren{x_{R\cup\{i\}}+b x_{R\cup\{j\}}}
    + b\Paren{x_{(R\symdiff f)\cup\{i\}} +b x_{(R\symdiff f)\cup\{j\}}} \\
    &= g_e(R)+b g_e(R\symdiff f).
\end{align*}
Therefore the total contribution of all edges disjoint from $e$ is
\[
    (L_{b,G-\{i,j\}}^{(k-1)}g_e)(R),
\]
where $G-\{i,j\}$ denotes the induced subgraph obtained by deleting the vertices $i$ and $j$.
Third, suppose that $f$ shares exactly one endpoint with $e$.
Let $f=\{i,a\}$ and let
\[ D_f=\Braces{Q\subseteq [n]\setminus\{i,a\}: \abs{Q}=k-1}.\]

We claim that there is a sign $\sigma_{e,f}\in\{\pm1\}$ and a bijection
$\pi_{e,f}:D_e\to D_f$ such that
\begin{equation}
\label{eq:twist-contribution}
    (L_{b,f}^{(k)}x)_{R\cup\{i\}} +b (L_{b,f}^{(k)}x)_{R\cup\{j\}}=
    \sigma_{e,f}\, g_f(\pi_{e,f}(R)).
\end{equation}
If $a\notin R$, the token can move along $f$ from $i$ to $a$ in $R\cup\{i\}$, while
there is no feasible move across $f$ from $R\cup\{j\}$. Hence
\[ (L_{b,f}^{(k)}x)_{R\cup\{i\}} +b (L_{b,f}^{(k)}x)_{R\cup\{j\}} = x_{R\cup\{i\}}+b x_{R\cup\{a\}},
\]
which is $\sigma_{e,f}g_f(R)$, depending only on the chosen orientation of $f$ and on
$b$.
Otherwise, $a\in R$. Then there is no feasible move across $f$ from $R\cup\{i\}$, while
the token at $a$ can move to $i$ from $R\cup\{j\}$. Thus
\[(R\cup\{j\})\symdiff\{i,a\}= (R\setminus\{a\})\cup\{i,j\},
\]
and
\begin{align*}
    (L_{b,f}^{(k)}x)_{R\cup\{i\}}
    +b (L_{b,f}^{(k)}x)_{R\cup\{j\}}
    &= b\Paren{x_{R\cup\{j\}}
    +b x_{(R\setminus\{a\})\cup\{i,j\}}} \\
    &= x_{(R\setminus\{a\})\cup\{i,j\}}
    +b x_{R\cup\{j\}},
\end{align*}
which is again
\[\sigma_{e,f}\,g_f((R\setminus\{a\})\cup\{j\}).\]

The map $\pi_{e,f}$ can therefore be stated explicitly as
\begin{equation*}
    R\mapsto
    \begin{cases}
        R, & a\notin R,\\
        (R\setminus\{a\})\cup\{j\}, & a\in R.
    \end{cases}
\end{equation*}
This is a bijection from $D_e$ to $D_f$.
The case $f=\{j,a\}$ is identical, with the roles of $i$ and $j$ interchanged.
Further, the twist term in \eqref{eq:twist-contribution} is just a relabelling of
the vertices and therefore preserves Euclidean norm:
\begin{equation}
\label{eq:twist-isometry}
    \Norm{g_f\circ \pi_{e,f}}_2=\Norm{g_f}_2.
\end{equation}
Combining the three cases, we obtain the derivative identity
\begin{equation}
\label{eq:derivative-identity}
    h_e = \Paren{2I+L_{b,G-\{i,j\}}^{(k-1)}}g_e +
    \sum_{\substack{f\in E(G)\\ \abs{f\cap e}=1}}
    \sigma_{e,f}\, g_f\circ \pi_{e,f}.
\end{equation}
We now take norms.
Let $d_i$ denote the degree of vertex $i$ in $G$.
The graph $G-\{i,j\}$ has $m-d_i-d_j+1$ edges, since deleting $i$ and $j$ removes
$d_i+d_j-1$ edges.
By the induction hypothesis,
\begin{equation*}
    \lambda_{\max}(L_{b,G-\{i,j\}}^{(k-1)})
    \leq
    m-d_i-d_j+1+(k-1).
\end{equation*}
Since $L_{b,G-\{i,j\}}^{(k-1)}$ is positive semidefinite, the operator norm of
$2I+L_{b,G-\{i,j\}}^{(k-1)}$ is
$2+\lambda_{\max}(L_{b,G-\{i,j\}}^{(k-1)})$.
Using \eqref{eq:derivative-identity} and \eqref{eq:twist-isometry}, we get
\begin{equation}
\label{eq:single-edge-bound}
    \Norm{h_e}_2 \leq \Paren{2+m-d_i-d_j+1+(k-1)}\Norm{g_e}_2 + \sum_{\substack{f\in E(G)\\ \abs{f\cap e}=1}}\Norm{g_f}_2.
\end{equation}
On the other hand, by \eqref{eq:h-g} and $\lambda\geq 0$,
\[\Norm{h_e}_2=\lambda\Norm{g_e}_2.\]
Summing \eqref{eq:single-edge-bound} over all edges $e=\{i,j\}$ gives
\begin{equation}
\label{eq:summed-bound}
    \lambda\sum_{e\in E(G)}\Norm{g_e}_2 \leq \sum_{e=\{i,j\}\in E(G)}
    \Paren{2+m-d_i-d_j+1+(k-1)}\Norm{g_e}_2+ \sum_{e\in E(G)} \sum_{\substack{f\in E(G)\\ \abs{f\cap e}=1}}\Norm{g_f}_2.
\end{equation}

We reindex the double sum. Fix an edge $f=\{u,v\}$.
The number of edges $e$ sharing exactly one endpoint with $f$ is
\[(d_u-1)+(d_v-1)=d_u+d_v-2.\]
Thus the total coefficient of $\Norm{g_f}_2$ on the right-hand side of
\eqref{eq:summed-bound} is
\[2+m-d_u-d_v+1+(k-1)+(d_u+d_v-2)=m+k.\]
Therefore,
\[\lambda\sum_{e\in E(G)}\Norm{g_e}_2 \leq (m+k)\sum_{e\in E(G)}\Norm{g_e}_2. \]
If $\sum_e\Norm{g_e}_2>0$, we divide by this quantity and obtain $\lambda\leq m+k$.
If $\sum_e\Norm{g_e}_2=0$, then $x_S+b x_T=0$ for every token-graph edge
$\{S,T\}$. Hence every local summand in $L_{b,G}^{(k)}x$ vanishes, so
$L_{b,G}^{(k)}x=0$ and $\lambda=0$.
In either case, $\lambda\leq m+k$.
Since $\lambda$ was an arbitrary eigenvalue of $L_{b,G}^{(k)}$, this proves the theorem.
\end{proof}

\section{Improved Bounds for Brouwer's Conjecture}

In this section, we prove \cref{thm:Brouwer}.

\paragraph{Background and notation.}

Brouwer's conjecture predicts that the error term $\epsilon_r(G)$ is at most
$\binom{r+1}{2}$ for every graph $G$ and every $1\leq r\leq n$.
We recast the partial sum $\sum_{i=1}^r\lambda_i(L(G))$ as the largest eigenvalue
of a natural operator on an exterior power.
For $U\subseteq[n]$, let $H_s(U)=\bigwedge\nolimits^s\R^U$, with the conventions
$H_0(U)=\R$ and $H_s(U)=\{0\}$ if $s<0$ or $s>\abs{U}$.
For an $s$-set $S=\{v_1<\cdots<v_s\}\subseteq U$, write
$e_S=e_{v_1}\wedge\cdots\wedge e_{v_s}$; the vectors $(e_S)_{\abs{S}=s}$
form an orthonormal basis of $H_s(U)$.

Let $M:\R^U\to\R^U$ be a linear operator. Its $s$-th additive compound is the operator
$M^{[s]}:H_s(U)\to H_s(U)$ defined on decomposable vectors by
\[
    M^{[s]}(v_1\wedge\cdots\wedge v_s)
    =
    \sum_{q=1}^s
    v_1\wedge\cdots\wedge Mv_q\wedge\cdots\wedge v_s.
\]
For $s=0$, set $M^{[0]}=0$ on $H_0(U)=\R$.
We write $A_s(G)=L(G)^{[s]}$ for the $s$-th additive compound of the Laplacian.
Thus, if $L(G)z_i=\lambda_i(L(G))z_i$, then
$z_{i_1}\wedge\cdots\wedge z_{i_s}$ is an eigenvector of $A_s(G)$ with eigenvalue
$\lambda_{i_1}(L(G))+\cdots+\lambda_{i_s}(L(G))$.
With the eigenvalues ordered as above, this gives
\begin{equation}
\label{eq:additive-spectrum}
    \lambda_{\max}(A_s(G))
    =
    \sum_{i=1}^s \lambda_i(L(G)).
\end{equation}
See for instance \cite{Lew26, Fiedler74, MarshallOA11} and the references therein for further exposition and properties of the additive compound. 

\paragraph{Proof.}

The proof follows the same edge-gradient template as the proof of
\cref{thm:token-lap-bound}.
Differentiating across an edge $e=\{i,j\}$ turns edges disjoint from $e$ into an
operator on $G-\{i,j\}$, while adjacent edges are handled by a twist term.
The exterior-power gradient $G_e$ has two components: $u_e$ is the signed analogue
of the token-graph edge derivative, and $w_e$ records the coordinates in which both
endpoints of $e$ are present.
We introduce an auxiliary operator $B_r$ so that the quadratic form of one edge is
$\Norm{u_e}_2^2+\Norm{w_e}_2^2$.

\paragraph{The edge operators.}

For an edge $e=\{i,j\}$, let $L_e=(e_i-e_j)(e_i-e_j)^\top$ and define
$A_{s,e}=L_e^{[s]}$.
This is the one-edge Laplacian operator, as in the previous section.
Let $W_{s,e}$ be the projection
\[
    W_{s,e}e_S=
    \begin{cases}
        e_S, & e\subseteq S,\\
        0, & e\not\subseteq S.
    \end{cases}
\]
Define $B_{s,e}=A_{s,e}-W_{s,e}$, $W_s(G)=\sum_{e\in E(G)}W_{s,e}$, and
$B_s(G)=\sum_{e\in E(G)}B_{s,e}$.
In particular, $A_s(G)=B_s(G)+W_s(G)$.
The operator $W_s(G)$ is diagonal in the exterior basis:
$W_s(G)e_S=e_G(S)e_S$, where $e_G(S)$ is the number of edges induced by $S$.
Hence,
\begin{equation}
\label{eq:W-bound}
    0\preceq W_s(G)\preceq \binom{s}{2}I.
\end{equation}

We first record a useful identity relating $B_{s,e}$ and $W_{s,e}$.

\begin{lemma}\label{lem:one-edge}
For every edge $e$ and every $s\geq 0$, the operator $B_{s,e}$ is positive semidefinite and
satisfies $B_{s,e}^2=2B_{s,e}-W_{s,e}$.
\end{lemma}

\begin{proof}
Since the single-edge operator only acts on the two coordinates $i$ and $j$, the space
$H_s([n])$ decomposes according to whether the indexing set $S$
contains neither endpoint of $e$, both endpoints of $e$, or exactly one endpoint of $e$.
We check the claim on each of these subspaces.

First, if $S$ contains neither endpoint of $e$, then $A_{s,e}e_S=0$, $W_{s,e}e_S=0$, and
$B_{s,e}e_S=0$.
Thus both sides of the identity vanish on this subspace.
Second, if $S$ contains both endpoints of $e$, then the one-edge Laplacian contributes
eigenvalue $2$ to the additive compound, so $A_{s,e}e_S=2e_S$.
Also $W_{s,e}e_S=e_S$, and hence
$B_{s,e}e_S=e_S$.
Here $B_{s,e}=W_{s,e}=I$, so $B_{s,e}^2=I=2B_{s,e}-W_{s,e}$.
It remains to consider the subspace spanned by
$e_{R\cup\{i\}}$ and $e_{R\cup\{j\}}$, where
$R\subseteq [n]\setminus\{i,j\}$ and $\abs{R}=s-1$.
Up to a possible change in signs from a suitable change of basis, $A_{s,e}=B_{s,e}$
has matrix
\[
    \begin{pmatrix}
        1 & -1\\
        -1 & 1
    \end{pmatrix},
\]
whose eigenvalues are $0$ and $2$.
On this subspace $W_{s,e}=0$.
Moreover, this matrix squares to twice itself, so again
$B_{s,e}^2=2B_{s,e}-W_{s,e}$.
Over all cases, the only eigenvalues of $B_{s,e}$ are $0,1,2$, so
$B_{s,e}\succeq 0$.
\end{proof}

Fix an orientation of every edge.
For an oriented edge $e=(i,j)$, set $V_e=[n]\setminus\{i,j\}$.
Let $\iota_i$ denote contraction by $e_i$, so that 
\[\iota_i(e_{v_1}\wedge\cdots\wedge e_{v_s})=
    \begin{cases}
        (-1)^{p-1}
        e_{v_1}\wedge\cdots\wedge \widehat{e_{v_p}}\wedge\cdots\wedge e_{v_s},
        & \text{if } i=v_p,\\
        0, & \text{if } i\notin\{v_1,\ldots,v_s\}.
    \end{cases}
\]
We also let $\pi_e^{(t)}:H_t([n])\to H_t(V_e)$ be the projection defined by
\[\pi_e^{(t)}e_S=
    \begin{cases}
        e_S, & S\subseteq V_e,\\
        0, & S\not\subseteq V_e.
    \end{cases}
\]
That is, $\pi_e^{(t)}$ keeps only the coordinates indexed by $t$-sets disjoint from the endpoints of $e$.
For $x\in H_r([n])$, define
\[
    u_e(x)=\pi_e^{(r-1)}(\iota_i-\iota_j)x\in H_{r-1}(V_e),
    \qquad
    w_e(x)=\pi_e^{(r-2)}\iota_j\iota_i x\in H_{r-2}(V_e).
\]
We define the ``edge gradient'' to be
\[
    G_e x=\Paren{u_e(x),w_e(x)}
    \in H_{r-1}(V_e)\oplus H_{r-2}(V_e).
\]
Thus $u_e$ is the signed version of the difference between the two coefficients indexed by
$R\cup\{i\}$ and $R\cup\{j\}$, while $w_e$ records the signed coefficient indexed by
$P\cup\{i,j\}$.
The adjoint of $G_e:H_r([n])\to H_{r-1}(V_e)\oplus H_{r-2}(V_e)$ is
$G_e^*:H_{r-1}(V_e)\oplus H_{r-2}(V_e)\to H_r([n])$, defined by
$\langle G_ex,(u,w)\rangle = \langle x,G_e^*(u,w)\rangle$.

\begin{lemma}
\label{lem:gradient-factor}
For every edge $e$, we have
\begin{equation}
\label{eq:gradient-factor}
    G_e^*G_e=B_{r,e}.
\end{equation}
Consequently,
\begin{equation}
\label{eq:gradient-sum}
    \sum_{e\in E(G)}\Norm{G_e x}_2^2
    =
    \langle x,B_r(G)x\rangle.
\end{equation}
\end{lemma}

\begin{proof}
As before, we split into the three cases determined by $S\cap e$.
First, if $S$ contains neither endpoint of $e$, then both components of $G_e$ vanish,
which agrees with $B_{r,e}=0$.
Second, if $S$ contains both endpoints, then $u_e$ vanishes and $w_e$ records that
coefficient with a sign, agreeing with $B_{r,e}=I$.

Finally, on the span of $e_{R\cup\{i\}}$ and $e_{R\cup\{j\}}$, the map $G_e$
becomes $\alpha-\epsilon\beta$, with $\epsilon\in\{\pm1\}$ determined by the ordering
in the exterior product.
The associated quadratic form is $(\alpha-\epsilon\beta)^2$, agreeing with the
quadratic form of $B_{r,e}$.
Thus $G_e^*G_e=B_{r,e}$, and summing over $e\in E(G)$ gives
\eqref{eq:gradient-sum}.
\end{proof}

For an edge $e=\{i,j\}$, write $G-\{i,j\}=G[[n]\setminus\{i,j\}]$, and let
$a_e=d_i+d_j-2$ be the number of edges adjacent to $e$.
For $x\in H_r([n])$, define
\[ T_r(x)= \sum_{\substack{e,f\in E(G)\\ e\neq f,\ e\cap f\neq\emptyset}}
    \langle x,B_{r,e}B_{r,f}x\rangle.\]
As we will see, this term is the analogue of the twist term in the previous section.

\begin{lemma}[Derivative identity]
\label{lem:derivative-identity}
Let $r\geq 2$, and let $x\in H_r([n])$. Then
\begin{align*}
    \sum_{e\in E(G)}
    \langle G_e x,G_e(B_r(G)x)\rangle
    &= \sum_{e=\{i,j\}\in E(G)} \left\langle  u_e(x), \Paren{2I+B_{r-1}(G-\{i,j\})}u_e(x) \right\rangle \\
    &\quad+ \sum_{e=\{i,j\}\in E(G)}\left\langle w_e(x), \Paren{I+B_{r-2}(G-\{i,j\})}w_e(x) \right\rangle + T_r(x).
\end{align*}
\end{lemma}

\begin{proof}
We fix $e=(i,j)$ and expand $B_r(G)=\sum_{f\in E(G)}B_{r,f}$.
First, if $f$ is disjoint from $e$, then $B_{r,f}$ acts only on coordinates in $V_e$, while the
contractions defining $u_e$ and $w_e$ remove $i$ and $j$. Hence, we have
\[
    u_e(B_{r,f}x)=B_{r-1,f}u_e(x),
    \qquad
    w_e(B_{r,f}x)=B_{r-2,f}w_e(x).
\]
Summing over such $f$ gives the operators $B_{r-1}(G-\{i,j\})$ and
$B_{r-2}(G-\{i,j\})$ respectively.

Second, if $f=e$, then
\[
    G_e(B_{r,e}x)=\Paren{2u_e(x),w_e(x)}.
\]
Finally, if $f$ is adjacent to $e$ and $f\neq e$, then by
\eqref{eq:gradient-factor},
\[
    \langle G_e x,G_e(B_{r,f}x)\rangle
    =
    \langle x,B_{r,e}B_{r,f}x\rangle.
\]
Summing over all $e$ and all adjacent $f\neq e$ gives $T_r(x)$.
\end{proof}

\paragraph{Bounding $T_r$ with stars.}

For a vertex $v$, write $E_v=\Braces{e\in E(G):v\in e}$ and
$S_v=\sum_{e\in E_v}B_{r,e}$.

\begin{lemma}\label{lem:star}
If $v$ has degree $d$, then $\lambda_{\max}(S_v)\leq d+\sqrt{r}$.
\end{lemma}

We note that a variant of this lemma was established in \cite[Lemma 5.1]{Lew26}, though we include its short proof here for completeness. 

\begin{proof}
Let $X$ be the set of the $d$ neighbors of $v$. It is clear that it suffices to analyze the exterior space generated by $\{v\}\cup X$.
Fix $q$ and consider the corresponding space 
\[ \bigwedge\nolimits^q\R^X
    \oplus
    \Paren{e_v\wedge \bigwedge\nolimits^{q - 1}\R^X}. \]
The first summand consists of configurations not containing the center $v$, and the second consists of configurations containing $v$. On $\bigwedge\nolimits^q\R^X$, each chosen leaf contributes one edge, so the
diagonal part is $qI$. On $e_v\wedge \bigwedge\nolimits^{q - 1}\R^X$, the diagonal part is $dI$.
Up to harmless sign changes of exterior basis vectors, the off-diagonal part corresponds to contraction by $\mathbf{1}_X=\sum_{x\in X}e_x$. Therefore, $S_v$ has block matrix
\[    M_q=
    \begin{pmatrix}
        qI & C_q^*\\
        C_q & dI
    \end{pmatrix},\]
where $C_q:\bigwedge\nolimits^q\R^X\to \bigwedge\nolimits^{q - 1}\R^X$ is contraction by $\mathbf{1}_X$, again up to changes of signed basis.
Since $\Norm{\mathbf{1}_X}_2=\sqrt{d}$ and contraction by a unit vector has operator norm at most $1$ on exterior space, we have $\Norm{C_q}\leq \sqrt{d}$. Thus, the Rayleigh quotient of $M_q$ is at most that of $\begin{pmatrix}
q & \sqrt{d}\\
\sqrt{d} & d
\end{pmatrix}$.  That is, for $q \in [d]$, we have 
\[\lambda_{\max}(M_q) \leq \frac{d+q+\sqrt{(d-q)^2+4d}}{2} \leq d+\sqrt{q}\leq d+\sqrt{r}.
\]
Finally, the edge cases: If $q=0$, the only eigenvalue is $0$; if $q=d+1$, then the eigenvalue is $d$, proving the desired bound.
\end{proof}

\begin{lemma}
\label{lem:Tr-estimate}
For every $x\in H_r([n])$,
\[T_r(x)
    \leq
    \sum_{e\in E(G)}
    (a_e+2\sqrt{r}-2)\Norm{u_e(x)}_2^2
    +
    \sum_{e\in E(G)}
    (a_e+2\sqrt{r})\Norm{w_e(x)}_2^2.\]
\end{lemma}

\begin{proof}
We rewrite $T_r$ by grouping terms according to their common vertex
\[
    T_r(x)
    =
    \sum_{v\in[n]}
    \left\langle
        x,
        \Paren{
            S_v^2-\sum_{e\in E_v}B_{r,e}^2
        }x
    \right\rangle.
\]
By \cref{lem:star}, $0\leq S_v\leq (d_v+\sqrt{r})I$, so that $S_v^2\leq (d_v+\sqrt{r})S_v.$
Combining this with \cref{lem:one-edge}, we have
\[S_v^2-\sum_{e\in E_v}B_{r,e}^2= S_v^2-2S_v+\sum_{e\in E_v}W_{r,e} \leq
    (d_v+\sqrt{r}-2)S_v+\sum_{e\in E_v}W_{r,e}.
\]
Taking the quadratic form at $x$ and summing over $v$ gives
\[ T_r(x)\leq\sum_{v\in[n]}(d_v+\sqrt{r}-2)
    \sum_{e\in E_v}\langle x,B_{r,e}x\rangle
    +
    \sum_{v\in[n]}\sum_{e\in E_v}\langle x,W_{r,e}x\rangle.
\]
By \cref{lem:gradient-factor}, $\langle x,B_{r,e}x\rangle = \Norm{u_e(x)}_2^2+\Norm{w_e(x)}_2^2$.
We also have $ \langle x,W_{r,e}x\rangle=\Norm{w_e(x)}_2^2$.
Combining all these gives the desired conclusion.
\end{proof}

\paragraph{Induction.} Define $\theta_0=0$ and
$\theta_r=\theta_{r-1}+2\sqrt{r}-1$ for $r\geq 1$.
Equivalently,
\begin{equation}
\label{eq:theta}
    \theta_r=2\sum_{s=1}^r\sqrt{s}-r.
\end{equation}

\begin{proposition}
\label{prop:inductive-step}
For every simple graph $G$ with $m$ edges and every $r\geq 1$, we have
\[\lambda_{\max}(B_r(G))\leq m+\theta_r.\]
\end{proposition}

\begin{proof}
The base case is $r=1$, where
$\lambda_{\max}(B_1(G))=\lambda_{\max}(L(G))\leq m+1=m+\theta_1$.
This follows for instance from \cite{AM85}. It suffices to prove the inductive step for $r \geq 2$. 
Let $\nu=\lambda_{\max}(B_r(G))$ and let $x\in H_r([n])$ be a corresponding unit
eigenvector $B_r(G)x=\nu x$.
Since $B_r(G)\succeq 0$, we have $\nu\geq 0$ and we may assume $\nu > 0$.
By \eqref{eq:gradient-sum}, since $G_e(B_r(G)x)=\nu G_e x$, we have
\begin{equation}
\label{eq:eigen-gradient}
    \nu\sum_{e\in E(G)}\Norm{G_e x}_2^2
    =
    \sum_{e\in E(G)}
    \langle G_e x,G_e(B_r(G)x)\rangle.
\end{equation}
For $e=\{i,j\}$, by the induction hypothesis applied to $G-\{i,j\}$, we have
\[
    \lambda_{\max}(B_{r-1}(G-\{i,j\}))\leq m-a_e-1+\theta_{r-1}, \quad\text{and}\quad
\lambda_{\max}(B_{r-2}(G-\{i,j\}))\leq m-a_e-1+\theta_{r-2}.
\]
Therefore, we have
\[\left\langle u_e, \Paren{2I+B_{r-1}(G-\{i,j\})}u_e \right\rangle \leq
    (m-a_e+1+\theta_{r-1})\Norm{u_e}_2^2,\]
and
\[\left\langle w_e, \Paren{I+B_{r-2}(G-\{i,j\})}w_e \right\rangle\leq (m-a_e+\theta_{r-2})\Norm{w_e}_2^2.\]
Combining these estimates with \cref{lem:derivative-identity} and \cref{lem:Tr-estimate}, we obtain
\begin{align*}
    \nu\sum_e\Norm{G_e x}_2^2 &\leq \sum_e
    \Paren{ m-a_e+1+\theta_{r-1} +a_e+2\sqrt{r}-2}\Norm{u_e}_2^2 \\
    &\quad + \sum_e \Paren{ m-a_e+\theta_{r-2}+a_e+2\sqrt{r}}\Norm{w_e}_2^2.
\end{align*}
The coefficient of $\Norm{u_e}_2^2$ is $m+\theta_{r-1}+2\sqrt{r}-1=m+\theta_r$,
while the coefficient of $\Norm{w_e}_2^2$, since $r \geq 2$, is
$m+\theta_{r-2}+2\sqrt{r}\leq m+\theta_{r}$.
Plugging into \eqref{eq:eigen-gradient}, we get
\[\nu\sum_e\Norm{G_e x}_2^2\leq(m+\theta_{r}) \sum_e\Paren{\Norm{u_e}_2^2+\Norm{w_e}_2^2}= (m+\theta_{r})\sum_e\Norm{G_e x}_2^2,
\]
which implies that $\nu\leq m+\theta_{r}$. This completes the induction.
\end{proof}

\begin{proof}[Proof of \cref{thm:Brouwer}]
Recall that $A_r(G)=B_r(G)+W_r(G)$.
By \eqref{eq:W-bound} and \cref{prop:inductive-step}, we have
\[ \lambda_{\max}(A_r(G))
    \leq
    \lambda_{\max}(B_r(G))+\binom{r}{2}
    \leq
    m+\theta_{r}+\binom{r}{2} = m+ \binom{r+1}{2}+ 2\sum_{s=1}^r(\sqrt{s}-1),
\]
where the last equality is due to \eqref{eq:theta}. The desired conclusion follows from \eqref{eq:additive-spectrum}.
The final numerical bound in the theorem follows from
$\sum_{s=1}^r\sqrt{s}\leq \int_0^r\sqrt{x}\,dx+\sqrt{r}
=\frac{2}{3}r^{3/2}+\sqrt{r}$ and $\sqrt{r}\leq r$.
\end{proof}

\section*{Acknowledgements}

A. Bakshi would like to thank Kunal Marwaha and Sidhanth Mohanty for helpful discussions in the early stages of this project. Part of this work was done while A. Bakshi was visiting the Simons Institute. We would also like to thank anonymous reviewers for helpful comments. 

\begingroup
\printbibliography

@inproceedings{parekh2021application,
  author       = {Ojas Parekh and Kevin Thompson},
  title        = {Application of the Level-2 Quantum Lasserre Hierarchy in Quantum Approximation
                  Algorithms},
  booktitle    = {48th International Colloquium on Automata, Languages, and Programming
                  ({ICALP} 2021)},
  pages        = {102:1--102:20},
  series       = {Leibniz International Proceedings in Informatics ({LIPIcs})},
  volume       = {198},
  editor       = {Nikhil Bansal and Emanuela Merelli and James Worrell},
  publisher    = {Schloss Dagstuhl -- Leibniz-Zentrum f{\"{u}}r Informatik},
  address      = {Dagstuhl, Germany},
  year         = {2021},
  url          = {https://drops.dagstuhl.de/entities/document/10.4230/LIPIcs.ICALP.2021.102},
  doi          = {10.4230/LIPIcs.ICALP.2021.102}
}

@inproceedings{anshu2020beyond,
  author       = {Anurag Anshu and David Gosset and Karen Morenz},
  title        = {Beyond Product State Approximations for a Quantum Analogue of Max Cut},
  booktitle    = {15th Conference on the Theory of Quantum Computation, Communication
                  and Cryptography ({TQC} 2020)},
  pages        = {7:1--7:15},
  series       = {Leibniz International Proceedings in Informatics ({LIPIcs})},
  volume       = {158},
  editor       = {Steven T. Flammia},
  publisher    = {Schloss Dagstuhl -- Leibniz-Zentrum f{\"{u}}r Informatik},
  address      = {Dagstuhl, Germany},
  year         = {2020},
  url          = {https://drops.dagstuhl.de/entities/document/10.4230/LIPIcs.TQC.2020.7},
  doi          = {10.4230/LIPIcs.TQC.2020.7},
  eprint       = {2003.14394},
  archivePrefix = {arXiv},
  primaryClass = {quant-ph}
}

@misc{Feige08,
  author={Feige, Uriel},
  title={Small Linear Dependencies for Binary Vectors of Low Weight},
  year={2008},
  note={Manuscript},
  url={https://www.wisdom.weizmann.ac.il/~feige/mypapers/evencover.pdf}
}

@inproceedings{HsiehKM23,
  title={A Simple and Sharper Proof of the Hypergraph Moore Bound},
  author={Hsieh, Jun-Ting and Kothari, Pravesh K. and Mohanty, Sidhanth},
  booktitle={Proceedings of the 2023 Annual ACM-SIAM Symposium on Discrete Algorithms (SODA)},
  pages={2324--2344},
  year={2023},
  publisher={SIAM},
  doi={10.1137/1.9781611977554.ch89},
  url={https://doi.org/10.1137/1.9781611977554.ch89},
  eprint={2207.10850},
  archivePrefix={arXiv},
  primaryClass={math.CO}
}

@inproceedings{GuruswamiKM22,
  title={Algorithms and Certificates for Boolean CSP Refutation: Smoothed Is No Harder than Random},
  author={Guruswami, Venkatesan and Kothari, Pravesh K. and Manohar, Peter},
  booktitle={Proceedings of the 54th Annual ACM SIGACT Symposium on Theory of Computing},
  pages={678--689},
  year={2022},
  publisher={Association for Computing Machinery},
  doi={10.1145/3519935.3519955},
  url={https://doi.org/10.1145/3519935.3519955},
  eprint={2109.04415},
  archivePrefix={arXiv},
  primaryClass={cs.CC}
}

@article{KothariM23,
  title={An Exponential Lower Bound for Linear 3-Query Locally Correctable Codes},
  author={Kothari, Pravesh K. and Manohar, Peter},
  journal={arXiv preprint arXiv:2311.00558},
  year={2023},
  doi={10.48550/arXiv.2311.00558},
  url={https://arxiv.org/abs/2311.00558},
  eprint={2311.00558},
  archivePrefix={arXiv},
  primaryClass={cs.CC}
}

@article{PiddockM17,
  author       = {Stephen Piddock and
                  Ashley Montanaro},
  title        = {The complexity of antiferromagnetic interactions and 2D lattices},
  journal      = {Quantum Inf. Comput.},
  volume       = {17},
  number       = {7{\&}8},
  pages        = {636--672},
  year         = {2017},
  url          = {https://doi.org/10.26421/QIC17.7-8-6},
  doi          = {10.26421/QIC17.7-8-6},
  bibsource    = {dblp computer science bibliography, https://dblp.org}
}

@inproceedings{HwangNP0W23,
  author       = {Yeongwoo Hwang and
                  Joe Neeman and
                  Ojas Parekh and
                  Kevin Thompson and
                  John Wright},
  editor       = {Nikhil Bansal and
                  Viswanath Nagarajan},
  title        = {Unique Games hardness of Quantum Max-Cut, and a conjectured vector-valued
                  Borell's inequality},
  booktitle    = {Proceedings of the 2023 {ACM-SIAM} Symposium on Discrete Algorithms,
                  {SODA} 2023, Florence, Italy, January 22-25, 2023},
  pages        = {1319--1384},
  publisher    = {{SIAM}},
  year         = {2023},
  url          = {https://doi.org/10.1137/1.9781611977554.ch48},
  doi          = {10.1137/1.9781611977554.ch48},
  bibsource    = {dblp computer science bibliography, https://dblp.org}
}

@article{Piddock25,
  author       = {Stephen Piddock},
  title        = {Quantum Max-Cut is {NP} hard to approximate},
  journal      = {CoRR},
  volume       = {abs/2510.07995},
  year         = {2025},
  url          = {https://doi.org/10.48550/arXiv.2510.07995},
  doi          = {10.48550/arXiv.2510.07995},
  eprinttype   = {arXiv},
  bibsource    = {dblp computer science bibliography, https://dblp.org}
}

@inproceedings{GharibianP19,
  author       = {Sevag Gharibian and
                  Ojas Parekh},
  editor       = {Dimitris Achlioptas and
                  L{\'{a}}szl{\'{o}} A. V{\'{e}}gh},
  title        = {Almost Optimal Classical Approximation Algorithms for a Quantum Generalization
                  of Max-Cut},
  booktitle    = {Approximation, Randomization, and Combinatorial Optimization. Algorithms
                  and Techniques, {APPROX/RANDOM} 2019, Massachusetts Institute of Technology,
                  Cambridge, MA, USA, September 20-22, 2019},
  series       = {LIPIcs},
  pages        = {31:1--31:17},
  publisher    = {Schloss Dagstuhl - Leibniz-Zentrum f{\"{u}}r Informatik},
  year         = {2019},
  url          = {https://doi.org/10.4230/LIPIcs.APPROX-RANDOM.2019.31},
  doi          = {10.4230/LIPIcs.APPROX-RANDOM.2019.31},
  bibsource    = {dblp computer science bibliography, https://dblp.org}
}

@article{ParekhT22,
  author       = {Ojas Parekh and Kevin Thompson},
  title        = {An Optimal Product-State Approximation for 2-Local Quantum Hamiltonians
                  with Positive Terms},
  journal      = {CoRR},
  volume       = {abs/2206.08342},
  year         = {2022},
  url          = {https://doi.org/10.48550/arXiv.2206.08342},
  doi          = {10.48550/arXiv.2206.08342},
  eprinttype   = {arXiv},
  eprint       = {2206.08342},
  archivePrefix = {arXiv},
  primaryClass = {quant-ph}
}

@inproceedings{Lee22,
  author       = {Eunou Lee},
  title        = {Optimizing Quantum Circuit Parameters via {SDP}},
  booktitle    = {33rd International Symposium on Algorithms and Computation
                  ({ISAAC} 2022)},
  pages        = {48:1--48:16},
  series       = {Leibniz International Proceedings in Informatics ({LIPIcs})},
  volume       = {248},
  publisher    = {Schloss Dagstuhl -- Leibniz-Zentrum f{\"{u}}r Informatik},
  year         = {2022},
  url          = {https://drops.dagstuhl.de/entities/document/10.4230/LIPIcs.ISAAC.2022.48},
  doi          = {10.4230/LIPIcs.ISAAC.2022.48}
}

@article{King23,
  doi = {10.22331/q-2023-11-09-1180},
  url = {https://doi.org/10.22331/q-2023-11-09-1180},
  title = {An {I}mproved {A}pproximation {A}lgorithm for {Q}uantum {M}ax-{C}ut on {T}riangle-{F}ree {G}raphs},
  author = {King, Robbie},
  journal = {{Quantum}},
  issn = {2521-327X},
  publisher = {{Verein zur F{\"{o}}rderung des Open Access Publizierens in den Quantenwissenschaften}},
  volume = {7},
  pages = {1180},
  month = {Nov},
  year = {2023}
}

@inproceedings{LeeP24,
  author       = {Eunou Lee and Ojas Parekh},
  title        = {An Improved Quantum Max Cut Approximation via Maximum Matching},
  booktitle    = {51st International Colloquium on Automata, Languages, and
                  Programming ({ICALP} 2024)},
  pages        = {105:1--105:11},
  series       = {Leibniz International Proceedings in Informatics ({LIPIcs})},
  volume       = {297},
  publisher    = {Schloss Dagstuhl -- Leibniz-Zentrum f{\"{u}}r Informatik},
  year         = {2024},
  url          = {https://drops.dagstuhl.de/entities/document/10.4230/LIPIcs.ICALP.2024.105},
  doi          = {10.4230/LIPIcs.ICALP.2024.105}
}

@article{HuberTPG24,
  author       = {Felix Huber and Kevin Thompson and Ojas Parekh and Sevag Gharibian},
  title        = {Second Order Cone Relaxations for Quantum Max Cut},
  journal      = {CoRR},
  volume       = {abs/2411.04120},
  year         = {2024},
  url          = {https://doi.org/10.48550/arXiv.2411.04120},
  doi          = {10.48550/arXiv.2411.04120},
  eprinttype   = {arXiv},
  eprint       = {2411.04120},
  archivePrefix = {arXiv},
  primaryClass = {quant-ph}
}

@article{JorqueraKKSW24,
  author       = {Zackary Jorquera and Alexandra Kolla and Steven Kordonowy and
                  Juspreet Singh Sandhu and Stuart Wayland},
  title        = {Monogamy of Entanglement Bounds and Improved Approximation Algorithms
                  for Qudit Hamiltonians},
  journal      = {CoRR},
  volume       = {abs/2410.15544},
  year         = {2024},
  url          = {https://doi.org/10.48550/arXiv.2410.15544},
  doi          = {10.48550/arXiv.2410.15544},
  eprinttype   = {arXiv},
  eprint       = {2410.15544},
  archivePrefix = {arXiv},
  primaryClass = {quant-ph}
}

@article{JuN25,
  author       = {Nathan Ju and Ansh Nagda},
  title        = {Improved approximation algorithms for the {EPR} Hamiltonian},
  journal      = {CoRR},
  volume       = {abs/2504.10712},
  year         = {2025},
  url          = {https://doi.org/10.48550/arXiv.2504.10712},
  doi          = {10.48550/arXiv.2504.10712},
  eprinttype   = {arXiv},
  eprint       = {2504.10712},
  archivePrefix = {arXiv},
  primaryClass = {quant-ph}
}

@article{GriblingSS25,
  author       = {Sander Gribling and Lennart Sinjorgo and Renata Sotirov},
  title        = {Improved Approximation Ratios for the Quantum Max-Cut Problem on
                  General, Triangle-Free and Bipartite Graphs},
  journal      = {CoRR},
  volume       = {abs/2504.11120},
  year         = {2025},
  url          = {https://doi.org/10.48550/arXiv.2504.11120},
  doi          = {10.48550/arXiv.2504.11120},
  eprinttype   = {arXiv},
  eprint       = {2504.11120},
  archivePrefix = {arXiv},
  primaryClass = {quant-ph}
}

@inproceedings{ApteLMPS25,
  author       = {Anuj Apte and Eunou Lee and Kunal Marwaha and Ojas Parekh and
                  James Sud},
  title        = {Improved Algorithms for Quantum MaxCut via Partially Entangled
                  Matchings},
  booktitle    = {33rd Annual European Symposium on Algorithms ({ESA} 2025)},
  pages        = {101:1--101:14},
  series       = {Leibniz International Proceedings in Informatics ({LIPIcs})},
  volume       = {351},
  publisher    = {Schloss Dagstuhl -- Leibniz-Zentrum f{\"{u}}r Informatik},
  year         = {2025},
  url          = {https://drops.dagstuhl.de/entities/document/10.4230/LIPIcs.ESA.2025.101},
  doi          = {10.4230/LIPIcs.ESA.2025.101}
}

@article{ApteLMPSS25,
  author       = {Anuj Apte and Eunou Lee and Kunal Marwaha and Ojas Parekh and
                  Lennart Sinjorgo and James Sud},
  title        = {A 0.8395-approximation algorithm for the {EPR} problem},
  journal      = {CoRR},
  volume       = {abs/2512.09896},
  year         = {2025},
  url          = {https://doi.org/10.48550/arXiv.2512.09896},
  doi          = {10.48550/arXiv.2512.09896},
  eprinttype   = {arXiv},
  eprint       = {2512.09896},
  archivePrefix = {arXiv},
  primaryClass = {quant-ph}
}

@article{AptePS25,
  title={Conjectured Bounds for 2-Local Hamiltonians via Token Graphs},
  author={Anuj Apte and Ojas Parekh and James Sud},
  journal={ArXiv},
  year={2025},
  volume={abs/2506.03441},
  url={https://api.semanticscholar.org/CorpusID:279154857}
}

@book{BrouwerH12,
  author    = {Andries E. Brouwer and Willem H. Haemers},
  title     = {Spectra of Graphs},
  series    = {Universitext},
  publisher = {Springer},
  address   = {New York},
  year      = {2012},
  doi       = {10.1007/978-1-4614-1939-6},
  isbn      = {978-1-4614-1939-6}
}

@article{HaemersMTR10,
  author  = {Willem H. Haemers and Ali Mohammadian and Behruz Tayfeh-Rezaie},
  title   = {On the sum of Laplacian eigenvalues of graphs},
  journal = {Linear Algebra and its Applications},
  volume  = {432},
  number  = {9},
  pages   = {2214--2221},
  year    = {2010},
  doi     = {10.1016/j.laa.2009.03.038},
  url     = {https://doi.org/10.1016/j.laa.2009.03.038}
}

@article{Lew26,
  title={An Approximate Version of Brouwer's Laplacian Conjecture},
  author={Lew, Alan},
  journal={arXiv preprint arXiv:2601.17575},
  year={2026},
  doi={10.48550/arXiv.2601.17575},
  url={https://arxiv.org/abs/2601.17575},
  eprint={2601.17575},
  archivePrefix={arXiv},
  primaryClass={math.CO}
}

@article {AM85,
    AUTHOR = {Anderson, Jr., William N. and Morley, Thomas D.},
     TITLE = {Eigenvalues of the {L}aplacian of a graph},
   JOURNAL = {Linear and Multilinear Algebra},
  FJOURNAL = {Linear and Multilinear Algebra},
    VOLUME = {18},
      YEAR = {1985},
    NUMBER = {2},
     PAGES = {141--145},
      ISSN = {0308-1087,1563-5139},
   MRCLASS = {05C50},
  MRNUMBER = {817657},
MRREVIEWER = {Torrence\ D.\ Parsons},
       DOI = {10.1080/03081088508817681},
       URL = {https://doi.org/10.1080/03081088508817681},
}

@article {Fiedler74,
    AUTHOR = {Fiedler, Miroslav},
     TITLE = {Additive compound matrices and an inequality for eigenvalues
              of symmetric stochastic matrices},
   JOURNAL = {Czechoslovak Math. J.},
  FJOURNAL = {Czechoslovak Mathematical Journal},
    VOLUME = {24(99)},
      YEAR = {1974},
     PAGES = {392--402},
      ISSN = {0011-4642},
   MRCLASS = {15A42},
  MRNUMBER = {347858},
MRREVIEWER = {Marvin\ Marcus},
}

@book {MarshallOA11,
    AUTHOR = {Marshall, Albert W. and Olkin, Ingram and Arnold, Barry C.},
     TITLE = {Inequalities: theory of majorization and its applications},
    SERIES = {Springer Series in Statistics},
   EDITION = {Second},
 PUBLISHER = {Springer, New York},
      YEAR = {2011},
     PAGES = {xxviii+909},
      ISBN = {978-0-387-40087-7},
   MRCLASS = {26-02 (05-02 26D15 26D20 60E15)},
  MRNUMBER = {2759813},
       DOI = {10.1007/978-0-387-68276-1},
       URL = {https://doi.org/10.1007/978-0-387-68276-1},
}

@misc{KothariT26,
title={On Brouwer's Laplacian Conjecture},
author={Pravesh K. {Kothari} and Stefan Tudose},
year={2026},
eprint={2606.12197},
archivePrefix={arXiv},
primaryClass={math.CO},
url={https://arxiv.org/abs/2606.12197}
}
\endgroup

\end{document}